\documentclass{article}

\pdfoutput=1

\usepackage[utf8]{inputenc}
\usepackage[T1]{fontenc}

\usepackage{amsmath, amssymb, amsxtra,amsthm}
\usepackage{hyperref}

\usepackage{bookmark}
\usepackage{color,soul,xcolor}
\usepackage{mathtools}
\usepackage{dsfont}
\usepackage{framed}
\usepackage{comment}
\usepackage{thmtools}
\usepackage[bb=boondox]{mathalfa}
\usepackage{subcaption}
\usepackage{bigints}
\usepackage{tikz}
\usetikzlibrary{arrows,shapes,calc,fit,positioning, patterns, backgrounds}

\usepackage{booktabs,tabularx}
\usepackage{multicol}
\usepackage{tcolorbox}

\usepackage{rotating}

\usepackage{changes}
\definechangesauthor[color=green!40!black!70]{1}
\definechangesauthor[color=orange]{2}
\definechangesauthor[color=blue]{3}
\definechangesauthor[color=red]{imp}

\usepackage[english, ruled, onelanguage, linesnumbered, vlined]{algorithm2e}

\SetKwComment{com}{$\triangleright$\ }{}
\usepackage{todonotes}
\usepackage{centernot}
\usepackage{xparse}

\usepackage{pdfpages}

\definecolor{darkred}{RGB}{164,74,63}
\definecolor{lightblue}{HTML}{F8766D}
\definecolor{darkgreen}{HTML}{00BA38}
\definecolor{darkyellow}{HTML}{619CFF}

\colorlet{shadecolor}{lightgray!25}

\newrobustcmd*{\mycircle}[1]{\tikz{
	\node[draw=darkred, double, circle] (0,0) {};
	}
}
\newrobustcmd*{\myline}[1]{\tikz{\filldraw[draw=#1, line width = 3pt] (0,0) -- (0.5cm, 0);}}
\newrobustcmd*{\myfline}[1]{\tikz{\filldraw[draw=#1, line width = 0.2mm] (0,0) -- (0.5cm, 0);}}

\newtheorem{thm}{Theorem}[section]
\newtheorem{prop}[thm]{Proposition}
\newtheorem{lemma}[thm]{Lemma}

\newtheorem{definition}[thm]{Definition}
\newtheorem{rem}[thm]{Remark}

\newtheorem{exa}[thm]{Example}

\newtheoremstyle{definitionsty}{3pt}{3pt}{\slshape}{}{\bfseries}{.}{.5em}{}
\theoremstyle{definitionsty}
\newtheorem{tdefn}[thm]{Definition}

\newcommand{\cl}[1]{\mathcal{#1}}

\newcommand{\bigO}[1]{\mathop{}\!O{\left(#1\right)}}

\newcommand{\R}{\mathbb{R}}

\newcommand{\tfdis}{\textit{sf-}\textrm{metric}}
\newcommand{\tfdisp}{\textit{sf-}\textrm{metrics}}
\newcommand{\fdis}[2]{\mathit{d_f}\left(#1,#2\right)}

\newcommand{\tdis}[2]{\mathit{l}\left(#1,#2\right)}
\newcommand{\sdis}[2]{\mathit{d}(#1,#2)}

\newcommand{\remt}{\Omega}
\newcommand{\ev}{\remt\times V}

\newcommand{\starr}[2]{\cl{T}_{#1#2}}

\newcommand{\pat}[2]{#1\rightsquigarrow #2}

\newcommand{\dac}{\mathrm{R}}

\newcommand{\cac}{\cl{R}}

\newcommand{\ssmdg}{\mathrm{SSMD}_{\gamma}}

\newcommand*{\quot}[2]%
{\ensuremath{%
    #1/\!\raisebox{-.65ex}{\ensuremath{#2}}}}

\def\hili{
	\leavevmode\rlap{
		\hbox to \hsize{
			\color{yellow!50}
			\leaders\hrule height .8\baselineskip depth .5ex\hfill
		}
	}
}

\newcommand{\trstri}{\textrm{reachability triple}{}}
\newcommand{\trstrip}{\textrm{reachability triples}{}}

\DeclarePairedDelimiterX{\abs}[1]{\lvert}{\rvert}{\setargs{#1}}
\makeatletter

\DeclarePairedDelimiterX{\set}[1]{\{}{\}}{\setargs{#1}}
\NewDocumentCommand{\setargs}{>{\SplitArgument{1}{;}}m}
{\setargsaux#1}
\NewDocumentCommand{\setargsaux}{mm}
{\IfNoValueTF{#2}{#1} {#1\,\delimsize|\,\mathopen{}#2}}

\DeclarePairedDelimiterX{\lis}[1]{[}{]}{\lisargs{#1}}
\NewDocumentCommand{\lisargs}{>{\SplitArgument{1}{;}}m}
{\lisargsaux#1}
\NewDocumentCommand{\lisargsaux}{mm}
{\IfNoValueTF{#2}{#1} {#1\,\delimsize|\,\mathopen{}#2}}

\SetKwProg{Fn}{Function}{:}{}
\SetKwFunction{sfdis}{\textsc{SFD}}
\SetKwFunction{inner}{\textsc{SFDInner}}
\SetKwFunction{vpred}{\textsc{VPreds}}
\SetKwFunction{indcon}{\textsc{IndStream}}
\SetKwFunction{shortest}{\textsc{Shortest}}
\SetKwFunction{fastest}{\textsc{Fastest}}
\SetKwFunction{sfastest}{\textsc{Shortest-Fastest}}
\SetKwFunction{ESFP}{\textsc{ESFP}}
\SetKwFunction{cpred}{\textsc{Const-preds}}
\SetKwFunction{dfsrec}{\textsc{DFS-rec}}
\SetKwFunction{cldep}{\textsc{NetDep}}
\SetKwFunction{rect}{construct rectangles}
\SetKwFunction{dims}{SP dimensions}
\SetKwFunction{funtows}{construct\_towers}
\SetKwFunction{vsfp}{\textsc{VSFP}}
\SetKwFunction{nodemax}{\textsc{node\_maximal\_path}}
\SetKwFunction{forpas}{\textsc{forward\_pass}}
\SetKwFunction{backpas}{\textsc{backward\_pass}}

\begin{document}

\title{\huge{On computing distances and latencies in Link Streams}}

\author{
	Frédéric Simard\\
	School of Electrical Engineering and Computer Science \\
	University of Ottawa \\
	Ottawa, ON, Canada\\
	email: fsima063@uottawa.ca
}

\maketitle

\begin{abstract}
Link Streams were proposed a few years ago as a model of temporal networks. We seek to understand the topological and temporal nature of those objects through efficiently computing the distances, latencies and lengths of \emph{shortest fastest} paths. We develop different algorithms to compute those values efficiently. Proofs of correctness for those methods are presented as well as bounds on their temporal complexities as functions of link stream parameters. One purpose of this study is to help develop algorithms to compute centrality functions on link streams such as the betweenness centrality and the closeness centrality.
\end{abstract}

\thanks{
	A short version of this text is set to be presented at the \emph{International Conference on Advances in Social Networks Analysis and Mining (ASONAM '19)}, in Vancouver, Canada \cite{Simard2019b}.
}


\section{Introduction}

Network science has been greatly influenced in recent years by the notion of temporal networks. Researchers in various fields have observed that real data varies over time and that static networks are insufficient to capture the full extent of some phenomenon. Different models of temporal networks have been suggested, among which the \emph{Link Streams} of Latapy et al. \cite{Latapy2016a} that captures the network evolution in continuous time. As is the case with other forms of networks, the notions of paths and distances are fundamental to the study of link streams. Kempe et al. \cite{Kempe2002} mention the use of time-respecting paths to study temporal networks. They further mention applications to epidemiology, in which one would seek information about the spread of a virus in a population. Human interactions can also be analyzed with temporal networks as has been observed by Tang et al. \cite{Tang2010a} and the link stream framework can help advance those studies. Although online social networks can be thought to vary in discrete time, with tweets and retweets for example, in real social networks the interactions have durations which are important to take into account in order to have an accurate description of the data. To see how link streams can be used in practice, many studies have emerged from the SocioPatterns Collaboration that includes datasets on face-to-face contacts \cite{cattuto2010dynamics,sociopatterns} with temporal labels. Those datasets are valuable tools to more accurately investigate aspects of social networks such as homophily \cite{Stehle2013604} and epidemics \cite{Moinet2018}.

Latapy et al. develop the notion of \emph{shortest fastest paths} in their link stream model as a new concept of paths that gather together the temporal as well as the structural information of a link stream. A shortest fastest path is one that is shortest among the fastest paths between two endpoints. This type of path is used to define a \emph{betweenness} centrality and it appears other centrality functions could be so defined as well. A social network can thus be analyzed through different perspectives: using the \emph{distance} to measure how the connectivity of a group varies over time, the \emph{latency} to measure how quickly an information can spread into a group of people and the length of a shortest fastest path to measure how efficiently this information is relayed. Note also how the time a shortest path starts and ends influences the information it can spread. 

We propose here to compute the metrics of shortest (fastest) paths in a link stream with different algorithms. General definitions are presented in \autoref{sec:background}, followed by a state of the art on \autoref{sec:relwork}. Then, we present our two main methods in \autoref{sec:algorithms}, experiments in \autoref{sec:exps} and we conclude in \autoref{sec:conc}.

\section{Background}\label{sec:background}

Most definitions are taken from Latapy et al. \cite{Latapy2016a}. A link stream $L$ is a tuple $L=(T,V,E)$ where $T\subseteq \R$ is a set of time instants, $V$ is a finite set of nodes (vertices) and $E\subseteq T\times V\otimes V$ is a set of links (edges). Here, $V\otimes V$ denotes the set of unordered pairs of vertices and we write $uv\in V\otimes V$. We say an element $(t_v,v)\in T\times V$ is a temporal vertex.

An edge of $E$ is a tuple $(t,uv)$. Given an interval $I\subseteq T$, we write $(I,uv)\subseteq E$, instead of $I\times \set{uv}\subseteq E$, to mean all edges $(t,uv)$ such that $t\in I$ are in $E$. We say an edge $(I,uv)\subseteq E$ is maximal if there exists no other edge $(J,uv)\subseteq E$ such that $I\subset J$. We say a maximal edge $([a,b],uv)\subseteq E$ starts on $a$, ends on $b$ and has duration $b-a$. We let $\remt$ be the set of \emph{event times} of $T$, that is $\remt:=\set{t\in T ; \exists \mbox{ maximal edge } ([t,t'],uv)\subseteq E \mbox{ or }([t',t],uv)\subseteq E}$. Elements of $\remt\times V$ are called \emph{event nodes}. We write $E_\remt := \set{(t,uv)\in E; t\in \remt}$.

A maximal edge, as well as $\remt$ and $\ev$ are illustrated on the link stream of \autoref{fig:ls_simpleex}. On this link stream, $([1,2],cb)\subset E$ is a maximal edge, whereas $([1,1.5],cb)\subset E$ is not. Thus, $\remt = \set{0,1,2,3}$.
\begin{figure}
\begin{center}
\scalebox{0.8}{\begin{tikzpicture}[auto, thick, transform shape]
\tikzstyle{every node} = [ellipse, minimum size = 0pt]
\draw (-1,2) node(v20) {$c$};
\draw (4,2) node(v21) {}; 
\draw (-1,0) node(v00) {$a$};
\draw (4,0) node(v01) {}; 
\draw (-1,3) node(v30) {$d$};
\draw (4,3) node(v31) {}; 
\draw (-1,1) node(v10) {$b$};
\draw (4,1) node(v11) {}; 
\path[draw,dashed] (0,2) -- (v21); 
\path[draw,dashed] (0,0) -- (v01); 
\path[draw,dashed] (0,3) -- (v31); 
\path[draw,dashed] (0,1) -- (v11); 

\draw (0,3) node(v1) [fill=black] {};
\draw (0,2) node(v2) [fill=black] {};
\draw (0,1) node(v) {};
\draw (0,0) node(v) {};
\draw (1,3) node(v) {};
\draw (1,2) node(v3p) [fill=black] {};
\draw (1,1) node(v2p) [fill=black] {};
\draw (1,0) node(v) {};
\draw (2,3) node(v) {};
\draw (2,2) node(v3) [fill=black] {};
\draw (2,1) node(v4) [fill=black] {};
\draw (2,0) node(v) {};
\draw (3,3) node(v) [fill=black] {};
\draw (3,1) node(v5) [fill=black] {};
\draw (3,0) node(v6) [fill=black] {};
\draw (3,3) node(v7) [fill=black] {};
\draw (3,2) node(v8) [fill=black] {};

\path[draw] (v1) edge [bend left] (v2);
\path[draw] (v5) edge [bend left] (v6);
\draw (v7) edge [bend left] (v8);

\draw (v2p) edge [bend right] coordinate[midway](m2p) (v3p);
\path[opacity=0] (v3) edge [bend left] coordinate[midway](m3p) (v4);
\path[draw] (m2p) -- (m3p);

\path[draw, line width=3pt, darkgreen] (v1) edge [bend left] (v2);
\path[draw, line width=3pt, dashed, darkgreen] (v2) edge (v3p);
\path[draw, line width=3pt, darkgreen] (v3p) edge [bend left] (v2p);
\path[draw, line width=3pt, dashed, darkgreen] (v2p) edge (v4);
\path[draw, line width=3pt, dashed, darkgreen] (v4) edge (v5);
\path[draw, line width=3pt, darkgreen] (v5) edge [bend left] (v6);

\path[draw, line width=1.5pt, lightblue] (v1) edge [bend left] (v2);
\path[draw, line width=1.5pt, dashed, lightblue] (v2) edge (v3);
\path[draw, line width=1.5pt, lightblue] (v3) edge [bend left] (v4);
\path[draw, line width=1.5pt, dashed, lightblue] (v4) edge (v5);
\path[draw, line width=1.5pt, lightblue] (v5) edge [bend left] (v6);

\node[draw=darkred, double, fit = (v1)] {};
\node[draw=darkred, double, fit = (v6)] {};

\draw[thick] (0,-0.95) -- (0,-1.05) node[below] {$0$};
\draw[thick] (1,-0.95) -- (1,-1.05) node[below] {$1$};
\draw[thick] (2,-0.95) -- (2,-1.05) node[below] {$2$};
\draw[thick] (3,-0.95) -- (3,-1.05) node[below] {$3$};
\draw[thick] (4,-0.95) -- (4,-1.05) node[below] {$4$};
\draw (-1,-1) node(ut) {};
\draw (5,-1) node[label={[label distance = 0.5mm]225:$t$}](vt) {};
\path[->,draw] (ut) -- (vt);
\end{tikzpicture}}
\caption{
	A simple link stream with maximal edge $([1,2],cb)$.
}
\label{fig:ls_simpleex}
\end{center}
\end{figure}
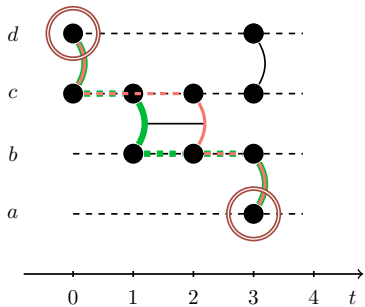

The graph $G_t$ induced by a time $t\in T$ is defined as $G_t = (V, \set{uv; (t,uv)\in E})$. In a link stream $L$, a path $P$ from $(\alpha,u)\in T\times V$ to $(\omega, v)\in T\times V$ is a sequence $(t_0,u_0,v_0), (t_1,u_1,v_1),\dots,(t_k,u_k,v_k)$ of elements of $T\times V\times V$ such that $u_0=u$, $v_k=v$, $t_0\geq \alpha$, $t_k\leq \omega$ and for all $i, t_i\leq t_{i+1}$, $v_i=u_{i+1}$ and $(t_i,u_iv_i)\in E$. We say that such a path starts at $t_0$, arrives at $t_k$, has length $k+1$ and duration $t_k-t_0$. We write $\pat{(\alpha,u)}{(\omega,v)}$ to mean that there exists a path from $(\alpha,u)$ to $(\omega,v)$ and say $(\omega,v)$ is reachable from $(\alpha,u)$. 
We also call $t_0$ a starting time and $t_k$ an arrival time from $(\alpha,u)$ to $(\omega,v)$. Each path between two fixed temporal nodes $(\alpha,u)$ and $(\omega,v)$ defines a pair of starting time and associated arrival time.
On the link stream of \autoref{fig:ls_simpleex}, two paths are illustrated: the green one $P_1 = (0,d,c),(1,c,b),(3,b,a)$ and the red one $P_2 = (0,d,c),(2,c,b),(3,b,a)$. Both have the same starting and arrival times from $(0,d)$ to $(3,a)$, namely times $0$ and $3$. Both paths are fastest.
We can also say $s$ is a starting time from a temporal node $(\alpha,u)\in T\times V$ to a node $v\in V$, in which case there exists some time $t\in T$ such that $s$ is the starting time of a path from $(\alpha,u)$ to $(t,v)$. Same goes for the arrival times. 

We say a path $P$ is shortest if it has minimal length and call its length the \emph{distance} from $(\alpha,u)$ to $(\omega,v)$, written $\sdis{(\alpha,u)}{(\omega,v)}$. Similarly, $P$ is fastest if it has minimal duration, in which case this duration is called the \emph{latency} from $(\alpha,u)$ to $(\omega,v)$ and is written $\tdis{(\alpha,u)}{(\omega,v)}$. Note that if $\pat{(\alpha,u)}{(\omega,v)}$, there exists at least one pair of starting time and arrival time $(s,a)\in \starr{(\alpha,u)}{(\omega,v)}$ such that $\tdis{(\alpha,u)}{(\omega,v)}=a-s$. 
Finally, $P$ is called shortest fastest if it has minimal length among the set of fastest paths from $(\alpha,u)$ to $(\omega,v)$. We call its length the \emph{\tfdis{}} from $(\alpha,u)$ to $(\omega,v)$ and write it $\fdis{(\alpha,u)}{(\omega,v)}$. In general, this is not a distance as it does not respect the triangular inequality and is only a premetric, a simple counterexample is shown on \autoref{fig:ls_ex_sfp}. On the same figure are drawn a shortest path, two fastest paths and a unique shortest fastest path.

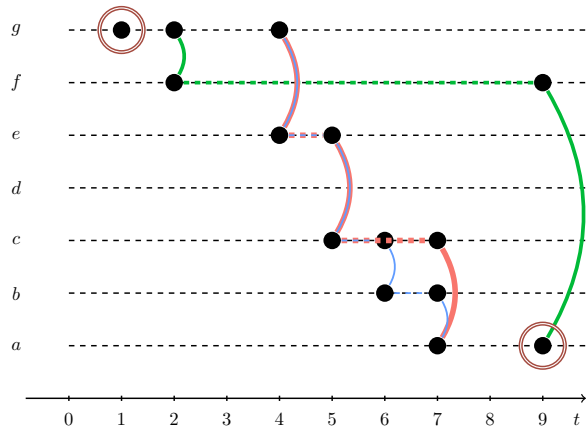
\begin{figure}
\begin{center}
\scalebox{0.7}{\begin{tikzpicture}[auto, thick, transform shape, scale = 1] 
\tikzstyle{every node} = [ellipse, minimum size = 0pt]
\draw (-1,3) node(v30) {$d$};
\draw (10,3) node(v31) {}; 
\draw (-1,5) node(v50) {$f$};
\draw (10,5) node(v51) {}; 
\draw (-1,2) node(v20) {$c$};
\draw (10,2) node(v21) {}; 
\draw (-1,0) node(v00) {$a$};
\draw (10,0) node(v01) {}; 
\draw (-1,6) node(v60) {$g$};
\draw (10,6) node(v61) {}; 
\draw (-1,1) node(v10) {$b$};
\draw (10,1) node(v11) {}; 
\draw (-1,4) node(v40) {$e$};
\draw (10,4) node(v41) {};
\path[draw,dashed] (0,3) -- (v31); 
\path[draw,dashed] (0,5) -- (v51); 
\path[draw,dashed] (0,2) -- (v21); 
\path[draw,dashed] (0,0) -- (v01); 
\path[draw,dashed] (0,6) -- (v61); 
\path[draw,dashed] (0,1) -- (v11); 
\path[draw,dashed] (0,4) -- (v41); 

\draw (1,6) node(v16) [fill=black] {};
\draw (2,6) node(v26) [fill=black] {};
\draw (2,5) node(v25) [fill=black] {};
\draw (4,6) node(v46) [fill=black] {};
\draw (4,4) node(v44) [fill=black] {};
\draw (5,4) node(v54) [fill=black] {};
\draw (5,2) node(v52) [fill=black] {};
\draw (6,2) node(v62) [fill=black] {};
\draw (6,1) node(v61) [fill=black] {};
\draw (7,2) node(v72) [fill=black] {};
\draw (7,1) node(v71) [fill=black] {};
\draw (7,0) node(v70) [fill=black] {};
\draw (9,5) node(v95) [fill=black] {};
\draw (9,0) node(v90) [fill=black] {};

\path[draw=darkgreen,line width=2pt] (v26) edge [bend left] (v25);
\path[draw=darkgreen,line width=2pt] (v95) edge [bend left] (v90);
\path[draw=darkgreen,line width=2pt,dashed] (v25) -- (v95);

\path[draw=lightblue,line width=3pt] (v46) edge [bend left] (v44);
\path[draw=lightblue,line width=3pt] (v54) edge [bend left] (v52);
\path[draw=lightblue,line width=3pt] (v72) edge [bend left] (v70);
\path[draw=lightblue,line width=3pt,dashed] (v44) -- (v54);
\path[draw=lightblue,line width=3pt,dashed] (v52) -- (v72);

\path[draw=darkyellow,line width=1pt] (v46) edge [bend left] (v44);
\path[draw=darkyellow,line width=1pt] (v54) edge [bend left] (v52);
\path[draw=darkyellow,line width=1pt] (v62) edge [bend left] (v61);
\path[draw=darkyellow,line width=1pt] (v71) edge [bend left] (v70);
\path[draw=darkyellow,line width=1pt,dashed] (v44) -- (v54);
\path[draw=darkyellow,line width=1pt,dashed] (v52) -- (v62);
\path[draw=darkyellow,line width=1pt,dashed] (v61) -- (v71);

\node[draw=darkred, double, fit = (v16)] {};
\node[draw=darkred, double, fit = (v90)] {};

\draw[thick] (0,-0.95) -- (0,-1.05) node[below] {$0$};
\draw[thick] (1,-0.95) -- (1,-1.05) node[below] {$1$};
\draw[thick] (2,-0.95) -- (2,-1.05) node[below] {$2$};
\draw[thick] (3,-0.95) -- (3,-1.05) node[below] {$3$};
\draw[thick] (4,-0.95) -- (4,-1.05) node[below] {$4$};
\draw[thick] (5,-0.95) -- (5,-1.05) node[below] {$5$};
\draw[thick] (6,-0.95) -- (6,-1.05) node[below] {$6$};
\draw[thick] (7,-0.95) -- (7,-1.05) node[below] {$7$};
\draw[thick] (8,-0.95) -- (8,-1.05) node[below] {$8$};
\draw[thick] (9,-0.95) -- (9,-1.05) node[below] {$9$};
\draw (-1,-1) node(ut) {};
\draw (10,-1) node[label={[label distance = 0.5mm]225:$t$}](vt) {};
\path[->,draw] (ut) -- (vt);

\end{tikzpicture}}
\caption{The shortest path from $(1,g)$ to $(9,a)$ (both encircled $\mycircle{darkred}$) is drawn in green $\myline{darkgreen}$. The two fastest paths are drawn in red $\myline{lightblue}$ and in blue $\myline{darkyellow}$. The sole shortest fastest path is the red one. Observe that, $\fdis{(1,g)}{(9,a)} = 3 > \fdis{(1,g)}{(9,f)} + \fdis{(9,f)}{(9,a)} = 2$.}
\label{fig:ls_ex_sfp}
\end{center}
\end{figure}

\section{Related work}\label{sec:relwork}

This work is close to the study of Wu et al. \cite{Wu2014}. As such, the applications of computing fastest and shortest paths mentioned by these authors also apply here. The main contribution of the present work is to compute \tfdisp{}, as well as distances and latencies, in a single pass over a dataset. Separately, Wu et al.'s fastest and shortest paths methods are insufficient to compute centralities such the betweenness of Latapy et al., while an algorithm combining them to produce \tfdisp{} is not efficient because it requires iterating multiple times over the dataset. Meanwhile, our methods iterate only once to produce the three metrics and are suitable for studying different aspects of a link stream. We also output information on the starting and arrival times of shortest (fastest) paths that give valuable information on connectivity. This study was instigated as a first step in computing Latapy et al.'s betweenness centrality. 

Furthermore, this work is also close to Tang et al. \cite{Tang2010} since these authors define a \emph{betweenness} centrality on temporal networks in terms of fastest shortest paths. Whether to use fastest shortest or shortest fastest paths (or any other path that combines temporal and structural information) depends on what information one wants to emphasize which depends on the context of the study. Shortest and fastest paths were also studied by Xuan et al. \cite{Xuan2003} and we were inspired by their all-pairs fastest path method to develop Algorithm \ref{alg:allpairs_sfp}. The latter is relevant to compute some centralities because metrics between all pairs of (temporal) nodes may be required. To our knowledge, Xuan et al.'s method is the only of its kind to return latencies between all pairs of nodes. More recently, Casteigts et al. \cite{Casteigts2015}. adopted the same strategies as Xuan et al. for computing shortest and fastest paths in a distributed way. 

Casteigts et al. \cite{Casteigts2013} also offer a survey of temporal networks that includes many applications of shortest and fastest paths. In particular, such paths can be used to study the reachability of a temporal node from another. It appears from that survey that either the distance or the latency is often used as a temporal metric to evaluate how well a temporal node can communicate with another. In this regard, the \tfdis{} can be used as another temporal function since it combines the temporal as well as the structural information into a single map. Note that the notion of \emph{foremost} paths (or journeys) is also used by some authors \cite{Casteigts2015} to study temporal reachability. A foremost path only has minimal arrival time, while its starting time is unconstrained. This type of path is also useful in many studies and we expect our algorithms can be extended to those cases to output lengths of \emph{shortest foremost} paths.

Finally, observe that the link stream framework is also close to the Time-Varying Graphs framework \cite{Casteigts2013}. Thus, all results presented in this paper carry to this other framework as well.

\section{Multiple-targets shortest fastest paths algorithms}\label{sec:algorithms}

The full implementations of the algorithms presented here, in C++, can be found online \cite{Simardgithub}.

We present here two main methods, Algorithms \ref{alg:forwpass} and \ref{alg:allpairs_sfp} that compute the distances, latencies and \tfdisp{} from one source event node to all other event nodes. Algorithm \ref{alg:allpairs_sfp} builds on the first method to compute those values for all pairs of event nodes. Subsection~\ref{sec:shortpatdelays} also presents Algorithm \ref{alg:sfp_special} that was derived from Algorithm \ref{alg:forwpass}. This last method was first devised to fairly compare Algorithm \ref{alg:forwpass} against the literature, but is also interesting as a standalone algorithm. We focus on the first two algorithms.

We present some small results that lead the way to those algorithms. The strategy for both methods is essentially the same: we compute the distances from any temporal node $(s_v,u)$ to $(t_v,v)$ such that $s_v$ is the \emph{largest} (or maximal) starting time from any $(t_u,u)$ to $(t_v,v)$. If it happens that $t_v - s_v = \tdis{(s_v,u)}{(t_v,v)}$, then this distance is the \tfdis{} from the former to the latter temporal node. Otherwise, since we iterate chronologically over $\remt$, this latency must have been computed at a time earlier than $t_v$ and is saved in memory. 

\subsection{Two simple lemmas}

The algorithms we present compute what we call \emph{\trstrip{}} that contain information about the lengths of shortest paths from one temporal node to another as well as the starting and arrival times of those paths.

\begin{definition}[Reachability triples]\label{def:shorttriple}
Let $(t_s,s)$ be an event node. If there exists a shortest path of length $l$ from $(t_s,s)$ to the event node $(t_y,y)$ that starts on a largest starting time $t\in \remt$, then we say $(t,t_y,l)$ is a \trstri{} from $(t_s,s)$ to $y$.
\end{definition}
In the following, we write $\dac_v$ for the dictionary of \trstrip{} from a fixed source event node to any node $v$. In order to reduce to cost of operations in $\dac_v$, we assume this dictionary is implemented in such a way that $\dac_v$ holds keys $s_v$ and $\dac_v[s_v]$ holds pairs $(a_v,d_v)$ that form \trstrip{} $(s_v,a_v,d_v)$. We write this dictionary so that accessing each $R_v$ takes constant time. 

Algorithms \ref{alg:forwpass} and \ref{alg:allpairs_sfp} compute distances \emph{from largest starting times only}. Those distances are contained in dictionaries $\dac_v$ for each $v\in V$ as part of \trstrip{}. Note that if a link stream reduces to a network, that is if the set of time instants $T$ is a singleton, then each $R_v$ will contain the usual distances from a fixed source to $v$. The temporal nature of a link stream forces us to take starting and arrival times into account when looking for shortest paths. Moreover, \trstrip{} could also be defined without the constraint that starting times are largest, however the algorithms would not be as efficient because the dictionaries would grow larger.

\autoref{lem:subpref} below, due to Wu et al. \cite{Wu2014}, states that shortest paths are prefix-shortest. We say a path $P_{(t_s,s)(t_u,u)}$ from a temporal node $(t_s,s)$ to another temporal node $(t_u,u)$ is a prefix of another path $P_{(t_s,s)(t_v,v)}$ from the same source to temporal node $(t_v,v)$ if $P_{(t_s,s)(t_u,u)}$ is a subsequence of $P_{(t_s,s)(t_v,v)}$. 

\begin{lemma}\label{lem:subpref}
Let $P_{(t_s,s)(t_v,v)}$ be a shortest path from a temporal node $(t_s,s)$ to another $(t_v,v)$. Then, every prefix $P_{(t_s,s)(t_u,u)}$ of $P_{(t_s,s)(t_v,v)}$ is a shortest path from $(t_s,s)$ to $(t_u,u)$.
\end{lemma}
\begin{proof}
Suppose otherwise and assume there exists a temporal node $(t_u,u)$ such that the prefix $P_{(t_s,s)(t_u,u)}$ of $P_{(t_s,s)(t_v,v)}$ is not shortest from $(t_s,s)$ to $(t_u,u)$. Then, there exists a shorter path from $(t_s,s)$ to $(t_u,u)$, $Q_{(t_s,s)(t_u,u)}$. Since $t_s\leq t_u \leq t_v$, we can use $Q_{(t_s,s)(t_u,u)}$ to form a shorter path to $(t_v,v)$, contradicting the minimality of $P_{(t_s,s)(t_v,v)}$.
\end{proof}

Let $(t_s,s)$ and $(t,v)$ be two temporal nodes. Then we define the outer distance from $(t_s,s)$ to $(t,v)$, $\sdis{(t_s,s)}{(t^-,v)}$, as either $\lim_{t_0\to t^-}\sdis{(t_s,s)}{(t_0,v)}$, when $t>t_s$, or $\sdis{(t_s,s)}{(t,v)}$, when $t_s=t$. \autoref{lem:bellmanshortpat} suggests it suffices to compute distances in induced graphs $G_t$ for any time $t$ to deduce the distances between two temporal nodes.

\begin{lemma}\label{lem:bellmanshortpat}
Let $(t_s,s)$ be a source temporal node and $(t_y,y)$ be a temporal node reachable from the source by a non-empty shortest path. Then, there exists $t_s\leq t\leq t_y$ and a connected component $C$ of $G_t$ such that
\begin{align}\label{eq:bellmanshortpat}
\sdis{(t_s,s)}{(t_y,y)} &= 
\min_{u,v\in C}
\sdis{(t_s,s)}{(t^-,u)}\\ 
&+ \sdis{(t,u)}{(t,v)}
+ \sdis{(t,v)}{(t_y,y)}.\nonumber
\end{align}
\end{lemma}
\begin{proof}
Let $P=(t_1,u_1,u_2), \dots, (t_n, u_n, u_{n+1})$ be a non-empty shortest path from $(t_s,s)$ to $(t_y,y)$. Then, $t_y\geq t_n \geq t_1 \geq t_s$ and $u_1 = s, u_{n+1}=y$. There exist non-empty subpaths in $P$ of the form $(t_j,u_j,u_{j+1}), \dots, (t_j, u_k, u_{k+1})$. Let $Q=(t_j,u_j,u_{j+1}), \dots, (t_j,u_k,u_{k+1})$ be such a subpath with the largest number of elements. By Lemma \ref{lem:subpref}, the prefix of $P$ from $(t_s,s)$ to $(t_j^-,u_j)$ is shortest and its length is $\sdis{(t_s,s)}{(t_j^-,u_j)}$. Moreover, the subpath of $P$ from $(t_j,u_{k+1})$ to $(t_y,y)$ must also be shortest with length $\sdis{(t_j,u_{k+1})}{(t_y,y)}$. Finally, since $P$ is shortest and the two subpaths formed by $P\setminus Q$ are shortest, $Q$ must also be a shortest path. Then, $Q$ has length $\sdis{(t_j,u_j)}{(t_j,u_{k+1})}$. The result follows by letting $t=t_j$ and $C = \set{u_j, u_{j+1}, \dots, u_{k+1}}$ be a connected component of $G_{t}$. 
\end{proof}

\subsection{A single-source method}

In this section, we present Algorithm \ref{alg:forwpass} that computes the distances (from largest starting times), latencies and \tfdisp{} from a source event node $(t_s,s)$ to all other reachable event nodes. This algorithm mixes iterations on the induced graphs $G_t$ for each time $t\in \remt$ with an all-pairs distances method on their connected components. Recall that if $s^*$ is the largest starting time from the source $(t_s,s)$ to some temporal node $(t,v)$, then either $t - s^* = \tdis{(t_s,s)}{(t,v)}$ or not. If so, then $\sdis{(t_s,s)}{(t,v)}$ is the \tfdis{} from $(t_s,s)$ to $(t,v)$. This length is computed with \autoref{lem:bellmanshortpat} by using the outer distances saved in memory as well as the all-pairs distance method on $G_t$. Thus, when we iterate over all pairs $(s_v,d_v)$ of starting time and outer distance from the source to $(t,v)$, we can deduce the duration and length of the shortest fastest paths from the source to $(t,v)$. This method uses a set $D$ that is assumed sorted in lexicographic order. Sorting $D$ helps lower the temporal complexity, but is not fundamental to understand the algorithm.

\begin{rem}
In Algorithms \ref{alg:forwpass} and \ref{alg:allpairs_sfp}, we assumed the dictionaries were implemented in the form of self-balanced binary trees in order to obtain logarithmic worst-case complexities. In our implementations, we used hash tables to lower the average-case complexity.
\end{rem}

Before proving that Algorithm \ref{alg:forwpass} is correct, let us go through a small example in order to build intuition. Algorithms \ref{alg:allpairs_sfp} and \ref{alg:sfp_special} are highly similar. 

\begin{exa}
Consider again the link stream of \autoref{fig:ls_ex_sfp}. Suppose the source is again $(1,g)$, $t=7$ and $C = \set{a,b,c}$. Thus, Algorithm \ref{alg:forwpass} will look for shortest (fastest) paths that can reach temporal nodes $(7,a),(7,b)$ and $(7,c)$. The unique largest starting time from the source to $C$ at time $7$ is $s_v = 4$. This time is given by the greatest key in $\dac_u$ for any $u\in C$. Then, we iterate over the outer distances from $(4,g)$ to $(7,v)$ for each $v\in C$. Note how the time of the source has changed from $1$ to $4$. By definition, and since the link stream is discrete, outer distances are given as the distances from $(4,g)$ to $(6,v)$ for each $v\in C$. Thus, we find outer distances $2$ from $(4,g)$ to $(7,c)$ and $3$ from $(4,g)$ to $(7,b)$. Node $a$ is discovered at time $7$ and its outer distance does not exist before that. Finally, combining the outer distances with the distances inside the graph induced by $C$ at time $7$, we find the distance from $(4,g)$ to $(7,c)$ is $2$, $3$ from $(4,g)$ to $(7,b)$ and also $3$ from $(4,g)$ to $(7,a)$. This last distance is given by the combination between the outer distance from $(4,g)$ to $(7,c)$ and the distance in $C$ from $(7,c)$ to $(7,a)$. Since node $a$ is discovered first at time $7$, that is its first arrival time from $(1,g)$ is $7$, then the latency from $(1,g)$ to $(7,a)$ is $\tdis{(1,g)}{(7,a)} = 7 - 4 = 3$ and the distance from $(1,g)$ to $(7,a)$ is the \tfdis{} from the former to the latter.
\end{exa}

\begin{prop}\label{prop:sfpcorr}
Algorithm \ref{alg:forwpass} correctly computes the latencies and \tfdisp{} from a source event node to all other reachable event nodes as well as the set of dictionaries $\set{\dac_v; v\in V}$. It requires at most $\bigO{\abs{V}^2\abs{\remt}^2\log{\abs{\remt}} + \abs{V}\abs{E_\remt}}$ operations in the worst case.
\end{prop}
\begin{proof}[Proof of correctness]
Let $(t_v,v)\in \remt\times V$ be some reachable destination. Let's show by induction on $\Delta:=\abs{\set{t_0\in \remt; t_v\geq t_0\geq t_s}}$ that $d[(t_v,v)] = \fdis{(t_s,s)}{(t_v,v)}$, $f[(t_v,v)]=\tdis{(t_s,s)}{(t_v,v)}$ and $\dac_v$ is correct up to time $t_v$.
\begin{itemize}
\item
When $\Delta = 1$, we iterate only on time $t_s$ and the result is clear.
\item
Suppose the result holds for all $k < \Delta$. Let $(t_1,\dots,t_{\Delta-1})$ be the times previously iterated over on line \ref{line:timeloop} and $t_\Delta$ the current time. By the induction hypothesis, by time $t_{\Delta-1}$, all values of $\dac_w$, for all $w\in V$, are correctly updated. Let $C_v$ be the connected component of $G_{t_\Delta}$ containing $v$. If $s\in C_v$, then the result follows as in the case with $\Delta=1$. Then, suppose $s\notin C_v$. Since each $\dac_v$ is correctly updated up to time $t_{\Delta-1}$ for each reachable $v\in V$, $D$ contains triples $(-s_w,d_w,w)$ for each $w\in C_v$ that have been visited prior to $t_{\Delta-1}$ from the source from a starting time $s_w$. The set $D$ contains the largest starting time $s_w$ from the source to $(t_{\Delta},w)$. Then, either $t_\Delta+s_w = \tdis{(t_s,s)}{(t_\Delta,w)}$ or this latency is given by some $f[(t_0,w)]$ such that $t_0<t_\Delta$. Let's iterate on $(-s_w,d_w,w)$. 

By Lemma \ref{lem:bellmanshortpat}, there exists a time $-s_w\leq t_i\leq t_{\Delta}$ and a connected component $C_i$ of $G_{t_i}$ such that $\sdis{(-s_w,s)}{(t_{\Delta},w)} = \min_{x,y\in C_i}\sdis{(-s_w,s)}{(t_i^-,x)} + \sdis{(t_i,x)}{(t_i,y)}$ $+ \sdis{(t_i,y)}{(t_{\Delta},w)}$. The sequence of distances 
\[
\sdis{(-s_w,s)}{(-s_w,u)},
\dots,
\sdis{(-s_w,s)}{(t_{\Delta},u)}
\] 
is non-increasing for each $u\in V$ because each element is minimal. Thus, since $w\in C_v$, in particular this lemma holds with $t_i=t_{\Delta}$ and $C_i = C_v$. Then, 
\begin{align*}
\sdis{(-s_w,s)}{(t_{\Delta},w)} 
&= 
\min_{x,y\in C_v}\sdis{(-s_w,s)}{(t_{\Delta}^-,x)}\\ 
&+ 
\sdis{(t_{\Delta},x)}{(t_{\Delta},y)} \\
&+ 
\sdis{(t_{\Delta},y)}{(t_{\Delta},w)}\\ 
&= 
\min_{u\in C_v} 
\sdis{(-s_w,s)}{(t_{\Delta-1},x)}\\ 
&+ 
\sdis{(t_{\Delta},x)}{(t_{\Delta},w)}.
\end{align*} 
By the induction hypothesis, the outer distance $d_x=\sdis{(-s_w,s)}{(t_{\Delta-1},x)}$ can be recovered from $(-s_w,t_{\Delta-1},d_x)\in \dac_x$ for each $x\in C_v$. Then, using $d_x$ and the dictionary $d'$ returned by the all-pairs distances algorithm on line \ref{line:apdist_sfpeff}, the expression above reduces to $\sdis{(-s_w,s)}{(t_{\Delta},w)} = \min_{x\in C_v}d_x + d'[(x,w)]$. In the last equation, the intermediary node $x\in C_v$ over which the minimum is taken is irrelevant. If $y\in U_x$, then the distance from the source to $y$ is the same as the distance from the source to $x$. Thus, it holds that:
\begin{align*}
\sdis{(-s_w,s)}{(t_{\Delta},w)} 
&= 
\min_{x\in C_v}
d_x + d'[(x,w)]\\
&=
\min_{x\in C_v}
\min_{y\in U_x}
d_y + d'[(y,w)]\\
&=
\min_{x\in C_v}
d_x +
\min_{y\in U_x}
d'[(y,w)].
\end{align*} 
Thus, when we iterate on the element $(-s_w,d_w,w)$ from $D$, we construct the set $U_w$ of nodes at distance $d_w$ from $(-s_w,s)$ at time $t_{\Delta}$. The last equation is thus used to insert into $\dac_w$ the right triple $(-s_w,t_{\Delta},\sdis{(-s_w,s)}{(t_{\Delta},w)})$ for each $w\in C_v$. When we have iterated over all of $D$, all dictionaries $\dac_v$ are correct at time $t_{\Delta}$. Finally, it suffices to observe that once $f[(t_\Delta,w)]$ is updated with its final value, then by definition the update of $d[(t_\Delta,w)]$ on \autoref{line:updiseff} yields the \tfdis{} from $(t_s,s)$ to $(t_\Delta,w)$ for each $w$. 
\end{itemize}
\end{proof}
\begin{proof}[Proof of complexity]
Let us write $n:=\abs{V}, m_t:=\abs{E_t}$ and $\omega:= \abs{\remt}$. On each time $t\in \set{t_0\in \remt; t_0\geq t_s}$, we first look up the connected components of $G_t$, which requires at most $\bigO{n + m_t}$ operations. On each component $C$ of $G_t$, we run an all-pairs distances method, which makes at most $\bigO{n^2 + nm_t}$ operations. For each node $v\in V$, the list in $\dac_v[-s_u]$ contains at most $ \omega$ elements since there can be at most as many pairs in $\dac_v[-s_u]$ as there are arrival times on $v$. The same goes for the number of keys $s_v$ in $\dac_v$.

There are at most $ \omega$ times $a_v$ such that $(s_v,a_v,d_v)\in \dac_v$ and thus $D$ can be constructed with at most $\bigO{n \omega}$ operations for all $v\in C$. Inserting and removing an element from $\dac_v[-s_u]$ takes at most $\bigO{1 + \log \omega + \log \omega}$ operations: $\bigO{1}$ operation for accessing $R_w$, $\bigO{\log \omega}$ operations for accessing key $-s_u$ and $\bigO{\log \omega}$ operations to insert or remove an item in a set of size at most $ \omega$. The costliest operations on the connected component $C$ are those insertions and deletions. Thus, operating over $C$ takes at most $\bigO{n\log \omega}$ operations. The list $D$ contains at most $\bigO{n \omega}$ triples since for each node $v\in V$, it holds a largest starting time and at most $ \omega$ distances (one distance for each arrival time on $v$). Thus, the \textbf{for} loop over $D$ will make at most $\bigO{n^2 \omega\log \omega}$ operations. 

The total number of operations at any time $t\in \remt$ is bounded above by $\bigO{n^2 \omega\log \omega} + \bigO{n^2 + nm_t}$. 
It suffices to multiply this sum by $\bigO{\omega}$ and use the observation that $\sum_{t\in \remt}m_t = \abs{E_\remt}$.
\end{proof}

Observe that we use the sets $V, E_{\remt}$ and $\remt$ as parameters to evaluate the temporal complexities of our algorithms. These appear as natural choices since $\remt$ indicates how the temporal dimension affects the number of operations while $E_{\remt}$ is a surrogate for $E$, which is in general infinite.

\subsection{A multiple-sources \tfdisp{} method}

Suppose $\remt$ is finite and starts on some time $a$. Algorithm \ref{alg:allpairs_sfp} returns a set of dictionaries of \tfdisp{} $D_{uv}$ for each pair of nodes $(u,v)\in V^2$ of dictionary $D_{uv}[s_{uv}] = (a_{uv}, d_{uv})$ such that $\tdis{(s_{uv},u)}{(a_{uv},v)} = a_{uv} - s_{uv}$ and $\sdis{(s_{uv},u)}{(a_{uv},v)} = d_{uv}$. During its execution, it updates a dictionary $D^0$ such that $D_{uv}[t] = (a_{uv},d_{uv})$, $t\in \dac_v$ and $(a_{uv}, d_{uv})\in \dac_v[t]$ from $(a,u)\in T\times V$. This dictionary helps in computing $D$ and in constructing $\dac_v$ from any source. It also returns a set of dictionaries $F_{uv}$ of latencies. 

\begin{prop}\label{prop:ap_sfpcorrect}
Algorithm \ref{alg:allpairs_sfp} returns the latencies, \tfdisp{} and dictionaries $\dac_v$ between all pairs of nodes in at most $\bigO{\abs{\remt}\abs{V}^2\left(\abs{V} + \abs{\remt}\right)\log{\abs{\remt}} + \abs{V}\abs{E_\remt}}$ operations.
\end{prop}
\begin{proof}[Proof of correctness]
Let us show that $D^0[u,v][t_v]$ holds correct \trstrip{} from $(a,u)$ to $(t_v,v)$ for any two nodes $u$, $v$ and time $t_v$. Thus, let us fix those three variables. Let us show this by induction on $\Delta := \abs{\set{t\in \remt; a\leq t \leq t_v}}$.
\begin{itemize}
\item 
If $\Delta = 1$, then either $u$ and $v$ are in the same connected component $C$ of $G_{t_v}$ or not. This part is clear.
\item
Suppose the result holds for any $k<\Delta$. Let $(t_1,\dots,t_{\Delta -1})$ be the sequence of times previously iterated over. Let $C_v$ be the connected component containing $v$ at time $t_\Delta$. If $u\in C_v$, then we argue as in the first case and the result follows. Otherwise, by the induction hypothesis, there must exist a largest starting time $s_v$ from $u$ to $(t_{\Delta},v)$ that can be found in $SA[u,w][t_{\Delta-1}]$, for some $w\in C_v$ since all such node $w$ is connected to $v$. Observe that $SA[u,v][t_{\Delta-1}]$ contains pairs of largest starting time and arrival time from $u$ to $(t_{\Delta-1},v)$. Observe also that $t_\Delta$ is again an arrival time on $v$. Thus, it suffices to compute the distance from $(s_v,u)$ to $(t_{\Delta},v)$ to obtain a \trstri{} $(s_v,t_{\Delta},d_v)$ from $(a,u)$ to $v$. We argue as in the proof of Algorithm \ref{alg:forwpass} that Algorithm \ref{alg:allpairs_sfp} returns this distance $d_v$. The update $D[u,v][t_{\Delta}][s^*] \gets d^*$ again follows the same reasoning as before.
\end{itemize}
\end{proof}

\begin{proof}[Proof of complexity]
Again, let $n:=\abs{V}, m_t:=\abs{E_t}$ and $\omega:= \abs{\remt}$. The costliest operations occur in the \textbf{for} loop starting on \autoref{line:forloopap}. There are at most $\omega$ keys on each $SA_{uv}$, for any $u,v\in V$. For any $t\in \remt$ and $u,v\in V$, the size of $SA_{uv}[t]$ is upper-bounded by $\omega$ since the starting time is maximal. Thus, at most $\bigO{1 + \log\omega + \log\omega} \subseteq \bigO{\log\omega}$ operations are required. Finding the largest starting time $s_v$ requires in the worst case $\bigO{n\log\omega}$ operations. By the same reasoning, the insertion on \autoref{line:megainsert} will make at most $\bigO{\omega\log\omega}$ operations.

$D^0_{uv}[t]$, for any $u,v\in V$ and $t\in \remt$, has a size at most $\omega^2$, thus the loop over $C$ to find $d_{\min}$ requires at most $\bigO{n\log\omega}$ operations. 

Recovering the last element of $D^0_{uv}[t]$ takes at most $\bigO{\log \omega}$ operations, thus the loop on $C_v$ makes at most $\bigO{\abs{C_v}\log{\omega}}$ operations. Meanwhile, inserting into $SA_{uv}[t]$ takes at most $\bigO{\omega \log{\omega}}$ operations. The \textbf{for} loop on \autoref{line:forloopap} thus makes at most:
\begin{align*}
\sum_{u\in C}\sum_{v\in V\setminus C}
\bigO{(\abs{C_v} + \omega)\log{\omega}}
&\leq
\sum_{u\in C}\sum_{v\in V}
\bigO{(\abs{C_v} + \omega)\log{\omega}}\\
&\leq
\sum_{u\in C}
\bigO{n(n + \omega)\log{\omega}}
\end{align*}
operations. This loop is itself repeated for all connected components $C\subseteq V(G_t)$, which in turn yields:
\begin{align*}
\sum_{C\subseteq V}\sum_{u\in C}
\bigO{n(n + \omega)\log{\omega}}
&=
\sum_{u\in V}
\bigO{n(n + \omega)\log{\omega}}\\
\end{align*}
operations.
Thus, this method should make at most
$
\bigO{n^2 + nm_t} +
\bigO{n^2(n + \omega)\log{\omega}}
$
operations in the worst case on each time $t$. This number of operations is repeated at most $\omega$ times and the result follows.
\end{proof}

Observe that Algorithm \ref{alg:forwpass} needs only be called $\abs{V}$ times in order to deduce the lengths of all shortest fastest paths from any source to any destination, since it discovers all starting times from each source. Thus, about $\bigO{\abs{\remt}^2\abs{V}^3\log\abs{\remt} + \abs{V}^2\abs{E_\remt}}$ operations are required for Algorithm \ref{alg:forwpass} to produce the same output as Algorithm \ref{alg:allpairs_sfp}. The multiple-sources algorithm is thus faster when the desired output is the set of \tfdisp{} from all sources to all destinations. The temporal complexities of both methods are affected mostly by the induced graphs $G_t$. In \autoref{sec:shortpatdelays}, we will see that complexities decrease drastically on cases such as $\gamma$-paths with $\gamma>0$ since we can remove the dependency on those induced graphs.

\subsection{Shortest paths with delays}\label{sec:shortpatdelays}

In \autoref{sec:wucomp}, we want to compare Algorithm \ref{alg:forwpass} against the shortest path procedure of Wu et al. \cite{Wu2014} on the same datasets they used. The shortest path procedure of these authors is the most efficient method known to return distances in temporal networks. However, this algorithm works only on paths with delays $\gamma>0$, that is $\gamma$-paths. 

A $\gamma$-path in a link stream is a path $(t_1,u_1,u_2),\dots,(t_n,u_n,u_{n+1})$ such that $t_i\geq t_{i-1} + \gamma$ for all $1<i\leq n$ and some $\gamma\in \R_+$. We call $\gamma$ the \emph{delay} and note that the usual path corresponds to a $0$-path. When $\gamma>0$, it is not necessary to iterate over connected components, since all nodes of a component do not communicate, and we can simplify Algorithm \ref{alg:forwpass} in order to reduce its number of operations. The complexities of algorithms \ref{alg:forwpass} and \ref{alg:allpairs_sfp} are mainly influenced by the operations related to the graphs $G_t$, for each time $t$, namely: finding connected components, computing the all-pairs distances and iterating on the set of nodes at equal distances in the connected component. When $\gamma>0$, we can remove the dependency on the induced graphs $G_t$ and accelerate our methods. Thus, we present Algorithm \ref{alg:sfp_special} that is deduced from Algorithm \ref{alg:forwpass} and assumes $\gamma>0$. Its correctness and temporal complexity follow from the same arguments used in \autoref{prop:sfpcorr}.
\begin{prop}
When $\gamma>0$, Algorithm \ref{alg:sfp_special} computes the latencies and \tfdisp{} from a source event node to all reachable event nodes as well as the set of dictionaries $\dac_v$, for all $v\in V$, in at most $\bigO{\abs{V} + \abs{E_\remt}\log\abs{\remt}}$ operations. 
\end{prop}
\begin{proof}
This follows from the same reasoning as in \autoref{prop:sfpcorr}. 
\end{proof}
Finally, in Algorithm \ref{alg:sfp_special}, the dictionaries $d$ and $f$ are implemented such that the keys are nodes and values are pairs $(t,k)$ such that $t$ is the time value $k$ is computed at that node. For example, if $(t,f_v)\in f[v]$, then the latency from the source to $(t,v)$ is $f_v$. This enables us to sort dictionaries by time. The same work could be done for Algorithm \ref{alg:allpairs_sfp}, that is to adapt it for the case $\gamma>0$, although that was not the focus here.

\section{Experiments}\label{sec:exps}

We present some experiments to highlight the running times of Algorithms \ref{alg:forwpass} and \ref{alg:allpairs_sfp}. In the first one, we compare Algorithm \ref{alg:sfp_special} with the  single-source shortest path method from Wu et al. \cite{Wu2014}. Algorithm \ref{alg:sfp_special} acts as a surrogate for Algorithm \ref{alg:forwpass}. Although Algorithm \ref{alg:allpairs_sfp} should be more efficient than Algorithm \ref{alg:forwpass} when the goal is to compute values between all pairs of temporal nodes, Wu et al. evaluated their method from a small set of source nodes on large datasets. It would be infeasible at this point to evaluate both our methods on the same datasets between all pairs of temporal nodes. In a second experiment, we compared the running times of our two methods on synthetic link streams. 

Algorithm \ref{alg:allpairs_sfp} was inspired by Xuan et al.'s fastest paths method that does not return distances. Comparing the two methods would be unfair against ours.

All experiments were run on a single machine with $2.6$ GHz Intel Core i7 processor and $16$ Gb of RAM. All methods were implemented in \textrm{C++} with standard libraries, including Wu et al.'s method. We implemented standard approaches to compute connected components and all pairs distances in graphs.

\subsection{Runtime comparison with the literature}\label{sec:wucomp}

We presented Algorithm \ref{alg:sfp_special} in \autoref{sec:shortpatdelays} that was motivated by a similar method developed by Wu et al. \cite{Wu2014}. We now compare how Algorithm \ref{alg:sfp_special} fares against their algorithm. Since we are not aware of methods comparable to Algorithms \ref{alg:forwpass} and \ref{alg:allpairs_sfp}, this is our comparison with the literature.

Wu et al. analyzed their method with the framework of temporal graphs and deduce a temporal complexity that is hard to compare with ours. We translate their result with link stream parameters, upper bounding $M$ with $\abs{E_{\remt}}$ and $d_{\max}$ with $\abs{\remt}$. Thus, the shortest path algorithm of Wu et al. makes at most $\bigO{\abs{V} + \abs{E_{\remt}}\log\abs{\remt}}$ operations in the worst case. The worst-case temporal complexities of both algorithms are thus the same.


We ran experiments on link streams of various sizes, as measured with $\abs{V}$, $\abs{\remt}$ and $\abs{E_{\remt}}$. We used the same datasets as Wu et al.\footnotemark, randomly chose $100$ different nodes from each and ran both methods one after the other. The full results (in seconds) can be found in \autoref{tab:runtime_comp_eff_wu}. The running times of Wu et al.'s method are either comparable or significantly less than that of Algorithm \ref{alg:sfp_special}. However, our method does more operations, since it must compute latencies as well and ensure the distances correspond to the \tfdisp{}. Thus, the running times of Wu et al.'s procedure are presented for reference only, it should not be expected that our methods would be faster. All datasets are heterogenous, which explains the variability in running times and we have not yet pinpointed any hidden link stream parameter that might explain this variability. The dictionaries $\dac_v$ are sensitive to the number of arrival times from the source and we suspect that in the problematic datasets some nodes must have a really high number of arrival times. This would make it more difficult to search values in some dictionary $\dac_v$.

\footnotetext{
	The datasets are only used as benchmarks. They all describe discrete temporal networks and can be found as part of the KONECT library of networks \cite{kunegis2013konect}. Only the values of the parameters $\abs{V}, \abs{\remt}$ and $\abs{E_\remt}$ were extracted since only these were required for our experiments.
}

\begin{table}[ht]
\centering
\resizebox{\linewidth}{!}{
\begin{tabular}{llll|rrr}
  \toprule
Dataset & $\abs{V}$ & $\abs{\remt}$ & $\abs{E_\remt}$ & Wu et al. (s)& $\ssmdg$ (s)& ratio \\ 
  \midrule
arxiv & 28093 & 2337 & 4596803 & 1.30 & 170.00 & 130.77 \\ 
  digg & 30398 & 9125 & 87627 & 1.60 & 1.10 & 0.69 \\ 
  elec & 7118 & 90741 & 103675 & 0.71 & 2.90 & 4.08 \\ 
  enron & 87273 & 178721 & 1148072 & 5.20 & 85.00 & 16.35 \\ 
  epinions & 755760 & 501 & 13668320 & 41.00 & 40.00 & 0.98 \\ 
  facebook & 63731 & 204914 & 817035 & 10.00 & 8.90 & 0.89 \\ 
  flickr & 2302925 & 134 & 33140017 & 120.00 & 3700.00 & 30.83 \\ 
  slashdot & 51083 & 67327 & 140778 & 4.80 & 4.30 & 0.90 \\ 
  wikiconflict & 116836 & 215982 & 2917785 & 6.90 & 21.00 & 3.04 \\ 
  wiki & 1870709 & 2198 & 39953145 & 100.00 & 22000.00 & 220.00 \\ 
  youtube & 3223585 & 203 & 9375374 & 170.00 & 160.00 & 0.94 \\ 
   \bottomrule
\end{tabular}
}
\caption{Comparisons between Algorithms \ref{alg:sfp_special} and \cite{Wu2014}} 
\label{tab:runtime_comp_eff_wu}
\end{table}

\subsection{Comparison between algorithms \ref{alg:forwpass} and \ref{alg:allpairs_sfp}}

Algorithm \ref{alg:forwpass} and \ref{alg:allpairs_sfp} were run on a set of randomly generated link streams of size $\abs{V}$ ranging from $100$ to $165$, with increments of $5$, and repeated $5$ times. Although the link streams are small in scale, the running times are significant since we compute the distances from every source to every destination. The link streams were constructed by generating Erdös-Renyi graphs $G(n,p)$, with $n=\abs{V}$ and $p=0.7$. Then, on each edge $(u,v)$, we drew a time instant $t \in \set{0,1,\dots,7}$ uniformly at random and added both directed edges $(t,u,v)$ and $(t,v,u)$ to $E$. In this case, edges have no duration and the time instants are integers: this helps ensure the size of $\remt$ is fixed and small, so the running times scale only with $\abs{V}$. 

\autoref{fig:runcomp_effap} presents the results of this comparison. We observe that, as the number of nodes involved increases, the amount of time taken by Algorithm \ref{alg:forwpass} grows faster than that of Algorithm \ref{alg:allpairs_sfp}. This gives clear indication that this method is faster than Algorithm \ref{alg:forwpass}. \autoref{tab:runtime_comp} shows the mean running times (over all repetitions of the same experiment) of each algorithms on a link stream with a fixed number of nodes. In terms of scale, the MSMD method manages a link stream of $160$ nodes and about $18000$ edges (the size of $E_{\remt}$ is an average over all repetitions) in, on average, less than $50$ seconds. Its counterpart takes more than $15$ minutes for the same calculations. 
\begin{figure}
\includegraphics[width = \linewidth]{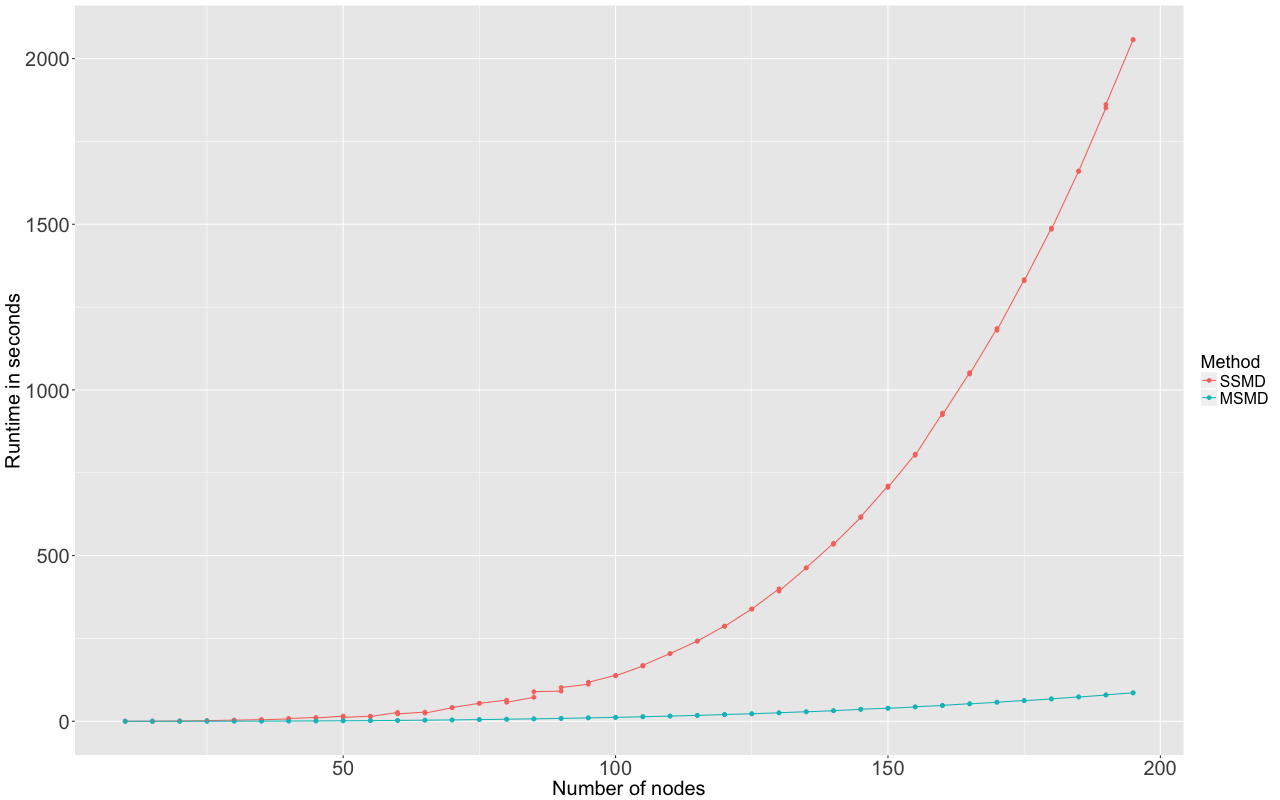}
\caption{Runtime comparison between Algorithms \ref{alg:forwpass} (SSMD) and \ref{alg:allpairs_sfp} (MSMD) on synthetic link streams (runtime in seconds vs number of nodes)}
\label{fig:runcomp_effap}
\end{figure}

Since Algorithm \ref{alg:allpairs_sfp} is more scalable than Algorithm \ref{alg:forwpass}, we generated a new set of link streams, again with the same process as before, although the time instants are now drawn uniformly at random in the interval $[0,10]$ while the duration of an edge $(t,uv)$ is drawn uniformly at random in the interval $[0,10 - t]$. Since $\remt$ grows on each generation, we kept $\abs{V}$ lower than in the former experiment and let $\abs{V} \in \set{10, 12, \dots, 68,70}$. The results are presented in the upper part of \autoref{tab:runtime_pred}, above the horizontal line with $\abs{V}$ up to $70$. We fitted, with the statistical software $\mathrm{R}$ \cite{Rsoftw}, a linear model on the runtime of Algorithm \ref{alg:allpairs_sfp} as function of both $\abs{V}$ and $\abs{\remt}$ in order to extrapolate the runtime of this method for larger values of $\abs{V}$ and $\abs{\remt}$. The fit is reasonable but imperfect, although this is sufficient to illustrate the scaling trend. Extrapolating, we obtain the values below the horizontal line. We observe that with around $190$ nodes and $12000$ event times, Algorithm \ref{alg:allpairs_sfp} should already take more than a day to finish. This suggests scalability might be an issue as we could not tackle a real-world dataset even with this long amount of time.

\begin{table}
\begin{subfigure}[ht]{0.45\linewidth}
\centering
\resizebox{\linewidth}{!}{
\begin{tabular}{ll|rr}
  \toprule
$\abs{V}$ & $\abs{E_\remt}$ & SSMD (s)& MSMD (s)\\ 
  \midrule
100 & 6942.40 & 142.36 & 11.74 \\ 
  105 & 7656.40 & 173.02 & 13.58 \\ 
  110 & 8394.00 & 208.28 & 15.47 \\ 
  115 & 9173.60 & 248.52 & 17.65 \\ 
  120 & 10005.60 & 293.48 & 20.01 \\ 
  125 & 10835.60 & 354.76 & 22.41 \\ 
  130 & 11723.20 & 404.31 & 25.23 \\ 
  135 & 12654.00 & 470.48 & 28.19 \\ 
  140 & 13601.60 & 547.13 & 31.50 \\ 
  145 & 14583.20 & 628.99 & 34.84 \\ 
  150 & 15609.20 & 718.28 & 38.40 \\ 
  155 & 16675.20 & 824.66 & 42.47 \\ 
  160 & 17794.80 & 946.35 & 46.74 \\ 
  165 & 18915.20 & 1107.04 & 52.12 \\ 
   \bottomrule
\end{tabular}
}
\caption{Comparisons between algorithms \ref{alg:forwpass} and \ref{alg:allpairs_sfp}} 
\label{tab:runtime_comp}
\end{subfigure}
\hfill
\begin{subfigure}[ht]{0.45\linewidth}
\centering
\resizebox{\linewidth}{!}{
\begin{tabular}{lll|r}
  \toprule
$\abs{V}$ & $\abs{\remt}$ & $\abs{E_\remt}$ & Runtime (s)\\ 
  \midrule
10 & 29 & 58 & 0.06 \\ 
  20 & 134 & 268 & 1.75 \\ 
  30 & 313 & 626 & 15.08 \\ 
  40 & 550 & 1100 & 64.20 \\ 
  50 & 857 & 1714 & 203.79 \\ 
  60 & 1225 & 2450 & 557.06 \\ 
  70 & 1670 & 3340 & 1248.19 \\ 
   \midrule
80 & 2177 & 4354 & 2379.12 \\ 
  100 & 3391 & 6783 & 6635.30 \\ 
  120 & 4872 & 9745 & 14814.50 \\ 
  140 & 6620 & 13241 & 28717.40 \\ 
  160 & 8635 & 17271 & 50469.20 \\ 
  180 & 10917 & 21835 & 82519.61 \\ 
  200 & 13466 & 26932 & 127642.88 \\ 
   \bottomrule
\end{tabular}
}

\caption{Runtimes of Algorithm \ref{alg:allpairs_sfp}} 
\label{tab:runtime_pred}
\end{subfigure}
\caption{Runtimes (in seconds) of Algorithms \ref{alg:forwpass} and \ref{alg:allpairs_sfp}}
\end{table}

\begin{algorithm}
\DontPrintSemicolon
\KwIn{$L=(T,V,E)$ a link stream, $\remt$ the set of event times, $(t_s,s)$ a source event node}
\KwOut{Dictionaries $d,f$ of \tfdisp{} and latencies from $(t_s,s)$ to all other event nodes, set of dictionaries $\dac_v$ for each $v\in V$}
$f, d \gets \mbox{create dictionaries}$\\
\lFor{$v\in V$}{
	$\dac_v \gets \mbox{create dictionary}$
}
\For{$t\in \mathrm{Sorted}(\set{t_0\in \remt; t_0\geq t_s})$}{ \label{line:timeloop}
	\For{$C\in \mathrm{connected\_components}(G_t)$}{\label{line:sfpd_eff_loop_conn}
		$H \gets G_t.\mathrm{induced\_subgraph}(C)$\\
		$d' \gets \mathrm{all\_pairs\_distances}(H)$\label{line:apdist_sfpeff}\\
		$D \gets \set{}$\\
		\lIf{$s\in C$}{
			$D.\mathrm{insert}(-t,0,s)$
		}
		\Else{
			$s_v \gets \max_{u\in C}\dac_u.\mathrm{last}()$\\
			\For{$v\in C $}{
				$D.\mbox{insert all $(-s_v,d_v,v)$ such that}$\\ 
				$(a_v,d_v)\in \dac_v[s_v]$ \mbox{for some $a_v$}
			}
		}
		\For{$(s_u,d_u,u) \in \mathrm{Sorted}(D)$}{\label{line:sfpd_eff_while}
			$U\gets \set{v\in C;\exists a_v : (a_v,d_u)\in \dac_v[-s_u]}$\\
			\lIf{$u = s$}{
				$U \gets \set{s}$
			}
			\For{$w\in C$}{
				$(\_,d_*)\gets \dac_w[s_u].\mathrm{last}()$\label{line:mindisteff}\\
				$d_{\min} \gets \min(d_u + \min_{u\in U}d'[(u,w)], d_*)$\\
				$\dac_w[-s_u].\mbox{remove all }(t,d_0)$ s.t. 
				$d_0>d_{\min}$\\
				$\dac_w[-s_u].\mathrm{insert}(t,d_{\min})$\label{line:dacinserteff}\\
				$f_w^*\gets \min_{(t_0,w)\in f}f[(t_0,w)]$\\
				$f[(t,w)] \gets \min(t+s_u,f_w^*)$\label{line:updlateff}\\
				$d[(t,w)] \gets$ $\min$ $d_0$ s.t. $s_0\in R_w$, $(a_0,d_0)\in R_w[s_0]$ and $a_0-s_0=f[(t,w)]$\label{line:updiseff}
			}
		}
	}
}
\Return{
	$d,f,\set{\dac_v; v\in V}$
}
\caption{SSMD \tfdis{}}
\label{alg:forwpass}
\end{algorithm}
\begin{algorithm}
\DontPrintSemicolon
\KwIn{$L=(T,V,E)$ a link stream, $\remt$ the set of event times}
\KwOut{$F$ a dictionary of latencies, $D^0$ a dictionary of \trstrip{}, $D$ a dictionary of \tfdisp{}}
\lFor{$u,v\in V$}{
	$SA_{uv}, F_{uv},D_{uv}, D^0_{uv}\gets$ create sorted dictionaries
}
\For{$t\in \remt$}{
	$t^- \gets $ last time of $\remt$ before $t$\\
	\For{$C\in \mathrm{connected\_components}(G_{t})$}{
		$H \gets G_{t}.\mathrm{induced\_subgraph}(C)$\\
		$d_C \gets \mathrm{all\_pairs\_distances}(H)$\\

		\For{$u,v\in C$}{
			$SA_{uv}[t].\mathrm{insert}(t,t)$\\
			$D^0_{uv}[t].\mathrm{insert}(t,t,d_C[u,v])$\\
			$F_{uv}[t].\mathrm{insert}(0)$
		}
		\For{$u\in C, v\in V\setminus C$}{\label{line:forloopap}
			$C_v \gets$ conn. component of $G_t$ containing $v$\\
			$s_v \gets \max_{w\in C_v, (s,a)\in SA_{uv}[t^-]}(s)$\\
			$SA_{uv}[t].\mathrm{insert}(s_v,t)$\\
			$SA_{uv}[t].\mathrm{insert}
			(\set{(s_v,a) \in SA_{uv}[t^-]})$\label{line:megainsert}\\
			$d_{\min}\gets \infty$\\
			\For{$w\in C_v$,}{
				$(\_,\_,d_w) \gets D^0_{uv}[t^-].\mathrm{last}()$\\
				$d_{\min} \gets \min(d_{\min}, d_w + d_C[w,v])$
			}
			$D^0_{uv}[t].\mathrm{insert}(s_v,t,d_{\min})$\\
			$l_{uv} \gets \min_{(s,a)\in SA_{uv}[t]}(a-s)$\\
			$l = \min(l_{uv}, F_{uv}[t^-])$\\
			$F_{uv}[t] \gets l$\\
			$(s^*, a^*) \gets $ pair $(s,a)\in SA_{uv}[t]$ s.t. $a - s = l$\\
			$D_{uv}[t][s^*] \gets d^*$ s.t. $(s^*,a^*,d^*)\in D^0_{uv}[t]$
		}
	}
}
\Return{$F,D^0,D$}
\caption{MSMD \tfdis{}}
\label{alg:allpairs_sfp}
\end{algorithm}

\begin{algorithm}
\DontPrintSemicolon
\KwIn{$L=(T,V,E)$ a link stream, $\remt$ the set of event times, $(t_s,s)$ a source event node}
\KwOut{Dictionaries $d,f$ of \tfdisp{} and latencies from $(t_s,s)$ to all other event nodes, set of dictionaries $\set{\dac_v; v\in V}$}
$d, f, \gets $ create dictionaries\\
\lFor{$v\in V$}{
	$\dac_v \gets$ create dictionary
}
\For{$(t,x,y)\in \mathrm{Sorted}(E)$ s.t. $t\geq t_s$}{
	\lIf{$u=s$}{
		$\dac_s[t].\mathrm{insert}(t,0)$
	}
	\If{$\dac_u\neq \emptyset$}{
		$s_u \gets \dac_u.\mathrm{last}()$\\
		$(a_u,d_u) \gets \dac_u[s_u].\mathrm{last}()$\\
		\If{$s_u$ exists}{
			$d_v \gets d_u + 1$\\
			\lIf{$\dac_v[s_u]$ does not contain $(t',d')$ s.t. $t'\leq t+\gamma$ and $d'<d_v$}{
				$\dac_v[s_u].\mathrm{insert}(t+\gamma,d_v)$
			}
		}
		$f_v \gets t - s_u$\\
		\If{$f[v] \neq \emptyset$}{
			$(\_, f_v') \gets f[v].\mathrm{last}()$\\
			\lIf{$f_v' < f_v$}{
				$f_v \gets f_v'$
			}
		}
		$f[v].\mathrm{add}(t,f_v)$\\
		$d_{\mathrm{fas}} \gets $ $\min$ $d_0$ s.t. $(a_0,d_0)\in R_v[s_u]$ and $a_0-s_u=f_v$\\
		\lIf{$d_{\mathrm{fas}}$ exists}{
			$d[v].\mathrm{add}(t, d_{\mathrm{fas}})$
		}
	}
}
\Return{
	$d,f,\set{\dac_v; v\in V}$
}
\caption{SSMD \tfdis{} with $\gamma>0$}
\label{alg:sfp_special}
\end{algorithm}

\section{Conclusion}\label{sec:conc}

In this paper, we presented three algorithms to compute metrics between pairs of event nodes. As opposed to similar known algorithms, those methods return all metrics at once in a single pass over the dataset. Moreover, the starting and arrival times of (some) shortest paths are returned, which is valuable information to compute, for example, the betweenness centrality of temporal nodes.

Algorithm \ref{alg:forwpass} works from a fixed source and is suitable when not all pairwise functions are required. Our experiments show that Algorithm \ref{alg:sfp_special}, and by extension Algorithm \ref{alg:forwpass}, is in general slower than the state of the art method to compute distances from a source node to all other nodes. However this is expected as it has to make more operations and work with bigger data structures. We did note some odd behaviour when comparing this method with the literature in that the ratio of running times between our method and Wu et al.'s method does not vary smoothly with known quantities. This should be inspected further if we would like to speed up the computation time of this method. Nevertheless, the focus of this study was to compute all metrics at once most efficiently, not to beat the state of the art distance method.

In practice, Algorithm \ref{alg:allpairs_sfp} has proved to finish its task faster than its counterpart on synthetic link streams. Since the link streams used were smaller than what we would expect from real-world instances, we extrapolated the running times produced by Algorithm \ref{alg:allpairs_sfp}. At this point, scalability is an issue and we could not expect to run this method on realistic link streams and obtain results in a reasonable amount the time. Thus, in order to speed up the computation time, we suggest studying how to lessen the amount of operations in either methods by skipping some temporal nodes and extrapolating the distances. Also, finding ways not to have to recompute the connected components and the all-pairs distances at every time would also be helpful in improving both methods. 

We believe the methods can be easily modified to compute other types of paths combining temporal and structural information, such as shortest foremost paths. In turn, those paths can be used to compute other centralities than the betweenness centrality or to investigate different topics such as reachability.

\bibliographystyle{plain}
\bibliography{ms}

\begin{thebibliography}{10}

\bibitem{sociopatterns}
Sociopatterns collaboration.
\newblock \url{www.sociopatterns.org/}.
\newblock Accessed: April 17, 2019.

\bibitem{Casteigts2015}
A.~Casteigts, P.~Flocchini, B.~Mans, and N.~Santoro.
\newblock {Shortest, fastest, and foremost broadcast in dynamic networks}.
\newblock {\em International Journal of Foundations of Computer Science},
  26(4):499--522, 2015.

\bibitem{Casteigts2013}
Arnaud Casteigts, Paola Flocchini, Walter Quattrociocchi, and Nicola Santoro.
\newblock Time-varying graphs and dynamic networks.
\newblock {\em International Journal of Parallel, Emergent and Distributed
  Systems}, 27(5):387--408, 2012.

\bibitem{cattuto2010dynamics}
Ciro Cattuto, Wouter Van~den Broeck, Alain Barrat, Vittoria Colizza,
  Jean-Fran{\c{c}}ois Pinton, and Alessandro Vespignani.
\newblock Dynamics of person-to-person interactions from distributed rfid
  sensor networks.
\newblock {\em PloS one}, 5(7):1--9, 07 2010.

\bibitem{Kempe2002}
David Kempe, Jon Kleinberg, and Amit Kumar.
\newblock {Connectivity and inference problems for temporal networks}.
\newblock {\em Journal of Computer and System Sciences}, 64(4):820--842, 2002.

\bibitem{kunegis2013konect}
J{\'e}r{\^o}me Kunegis.
\newblock Konect: the koblenz network collection.
\newblock In {\em Proceedings of the 22nd International Conference on World
  Wide Web}, pages 1343--1350. ACM, 2013.

\bibitem{Latapy2016a}
Matthieu Latapy, Tiphaine Viard, and Cl{\'e}mence Magnien.
\newblock Stream graphs and link streams for the modeling of interactions over
  time.
\newblock {\em Social Network Analysis and Mining}, 8(1):61, 2018.

\bibitem{Moinet2018}
Antoine Moinet, Romualdo Pastor-Satorras, and Alain Barrat.
\newblock Effect of risk perception on epidemic spreading in temporal networks.
\newblock {\em Phys. Rev. E}, 97:012313, Jan 2018.

\bibitem{Rsoftw}
{R Core Team}.
\newblock {\em R: A Language and Environment for Statistical Computing}.
\newblock R Foundation for Statistical Computing, Vienna, Austria, 2013.

\bibitem{Simardgithub}
Frédéric Simard.
\newblock \textsc{SSMD} and \textsc{MSMD} repository.
\newblock \url{https://bitbucket.org/simfr404/linkstreams_cpp/src/master/}.
\newblock Accessed: \today.

\bibitem{Simard2019b}
Frédéric Simard.
\newblock {On computing distances and latencies in Link Streams}.
\newblock In {\em Proceedings of The 2019 IEEE/ACM International Conference on
  Advances in Social Networks Analysis and Mining}, Vancouver, Canada, 2019.
  ACM.

\bibitem{Stehle2013604}
Juliette Stehlé, François Charbonnier, Tristan Picard, Ciro Cattuto, and
  Alain Barrat.
\newblock Gender homophily from spatial behavior in a primary school: A
  sociometric study.
\newblock {\em Social Networks}, 35(4):604 -- 613, 2013.

\bibitem{Tang2010a}
John Tang, Mirco Musolesi, Cecilia Mascolo, and Vito Latora.
\newblock {Characterising temporal distance and reachability in mobile and
  online social networks}.
\newblock {\em ACM SIGCOMM Computer Communication Review}, 40(1):118, 2010.

\bibitem{Tang2010}
John Tang, Mirco Musolesi, Cecilia Mascolo, Vito Latora, and Vincenzo Nicosia.
\newblock {Analysing Information Flows and Key Mediators through Temporal
  Centrality Metrics}.
\newblock In {\em Proceedings of the 3rd Workshop on Social Network Systems
  (SNS '10)}, Paris, France, 2010. ACM.

\bibitem{Wu2014}
Huanhuan Wu, James Cheng, Silu Huang, Yiping Ke, Yi~Lu, and Yanyan Xu.
\newblock {Path Problems in Temporal Graphs}.
\newblock {\em Proceedings of the VLDB Endowment}, 7(9):721--732, 2014.

\bibitem{Xuan2003}
B~Bui Xuan, Afonso Ferreira, and Aubin Jarry.
\newblock Computing shortest, fastest, and foremost journeys in dynamic
  networks.
\newblock {\em International Journal of Foundations of Computer Science},
  14(02):267--285, 2003.

\end{thebibliography}

\end{document}